\documentclass[runningheads]{llncs}

\usepackage{amsmath}
\usepackage{enumitem}
\usepackage{amssymb}
\usepackage{algorithm}
\usepackage{algorithmic}
\DeclareMathAlphabet\rsfscr{U}{rsfso}{m}{n}

\def \NP   	{\mathbf{NP}}
\def \FNP   	{\mathbf{FNP}}
\def \TFNP   	{\mathbf{TFNP}}

\def \IP        {\mathbf{IP}}
\def \CLS       {\mathbf{CLS}}
\def \PSPACE    {\mathbf{PSPACE}}
\def \PPAD      {\mathbf{PPAD}}
\def \VDF       {\mathbf{VDF}}
\def \TM    	{\mathsf{TM}}
\def \PTM    	{\mathsf{PTM}}
\def \O     	{\mathcal{O}}
\def \Z 	{\mathbb{Z}}

\def \L 	{\mathcal{L}}
\def \S 	{\mathsf{S}}
\def \P     	{\mathsf{P}}
\def \V 	{\mathsf{V}}
\def \M 	{\mathsf{M}}

\def \RO	{\mathsf{H}}
\def \adv       {\mathcal{A}}
\def \adf       {\mathcal{B}}

\def \pp   	{pp}
\def \setup	{\mathsf{Setup}}
\def \eval	{\mathsf{Eval}}
\def \vdf	{\textsc{VDF}}
\def \verify	{\mathsf{Verify}}

\def \ips	{(\prv \leftrightarrow  \vrf)}
\def \vdfs	{\langle \prv \leftrightarrow  \vrf \rangle}
\def \advs	{\langle \adv \leftrightarrow  \vrf \rangle}
\def \adfs	{\langle \adf \leftrightarrow  \vrf \rangle}
\def \negl	{\mathtt{negl}}
\def \poly	{\mathtt{poly}}
\def \st	{\mathtt{state}}
\def \mod       {\; \mathbf{mod} \;}
\def \multgroup#1{(\mathbb{Z}/#1\mathbb{Z})^\times}
\def \vrf       {\mathcal{V}}
\def \prv       {\mathcal{P}}
\def \X         {\mathcal{X}}
\def \Y         {\mathcal{Y}}

\def \EOL	{\textsc{EOL}}
\def \EOML	{\textsc{EOML}}
\def \SVL	{\textsc{SVL}}
\def \rSVL	{\textsc{rSVL}}

\begin{document}
\title{How hard are verifiable delay functions?}
\titlerunning{$\VDF \subsetneq \CLS$}
\author{Souvik Sur
\orcidID{0000-0003-1109-8595}}
%
%
\institute{
\email{souviksur@gmail.com}}
\maketitle              

\begin{abstract}
Verifiable delay functions ($\vdf)$ are functions that take
a specified number of sequential steps to be evaluated but can be verified efficiently. 
In this paper, we introduce a new complexity class that contains all the $\vdf$s. 
We show that this new class $\VDF$ is a subclass of $\CLS$ (continuous local search) and
\textsc{Relaxed-Sink-of-Verifiable-Line} is a complete problem for the class $\VDF$.

\end{abstract}


\section{Introduction}\label{introduction}
 In 1992, Dwork and Naor introduced the very first notion 
 of $\vdf$ under a different nomenclature ``pricing function" \cite{Dwork1992Price}.
 It is a computationally hard puzzle that needs to be solved
 to send a mail, whereas the solution of the puzzle can be verified efficiently. 
 Later, the concept of verifiable delay functions was
 formalized in~\cite{Dan2018VDF}.

Given the security
parameter $\lambda$ and delay parameter $T$, the prover needs 
to evaluate the $\vdf$ in time $T$. The verifier
verifies the output in $\poly(\lambda,\log T)$-time 
using some proofs produced by the prover. 
A crucial property of $\vdf$s, namely sequentiality,
ensures that the output can not be computed in time much less than $T$ even in the
presence of $\poly(\lambda,T)$-parallelism. 
$\vdf$s have several applications ranging
from non-interactive time-stamping to resource-efficient blockchains, however, are really
rare in practice because of the criteria sequentiality.
In order to design new $\vdf$s we must find problems that offer sequentiality. 
To the best of our knowledge so far, 
all the practical $\vdf$s are based on two inherently sequential algebraic problems 
-- modular exponentiation in groups of unknown
order~\cite{Pietrzak2019Simple,Wesolowski2019Efficient} (fundamentally known as the time-lock puzzle
\cite{Rivest1996Time}) and isogenies over super-singular curves~\cite{Feo2019Isogenie}. 
The security proofs of these $\vdf$s
are essentially polynomial-time reductions from one of these assumptions to the corresponding $\vdf$s.
Thus, from the perspective of the designers, the first hurdle is to find inherently sequential
problems. 
 
The main motivation behind this study has been where should we search for
such inherently sequential problems which are also efficiently verifiable?
In this paper, we show that the class of all $\vdf$s, namely $\VDF$, is a subclass of the
class $\CLS$ (continuous local search). In particular, we prove that
\textsc{Relaxed-Sink-of-Verifiable-Line}~\cite{ChoudhuriNash} is a complete problem for
this new class $\VDF$.

\subsection{Proof Sketch}\label{contributions}
The key challenges with this aim are, 
\begin{description}
\item [$\VDF$-hardness] reducing any arbitrary $\vdf$ into a hard distribution of $\rSVL$
(abbreviation for \textsc{Relaxed-Sink-of-Verifiable-Line}) instances.

\item [$\VDF$-membership] producing $\vdf$s from a family of subexponentially-hard $\rSVL$ instances.
 \end{description}

First, we define $\vdf$ as a language that helps us, in turn, to define the
class $\VDF$ as a special case of interactive proofs. 
The main idea is to detach the Fiat--Shamir transformation from the
traditional definition of $\vdf$ in order to find its hardness 
irrespective of any random oracle.

For the first task, Choudhuri et al. suggest a method in~\cite{ChoudhuriRSW} 
for the $\vdf$s that need proof. We show that it is possible even with the $\vdf$s that have no
proofs. The trick is to define the $\eval$ function of a $\vdf$ in terms of an iterated
sequential function $f$. Use this $f$ in order to design the successor circuit $\S$ of
$\rSVL$ instances. As a much easier case, we also show that the permutation $\vdf$s can be
reduced to such hard $\rSVL$ instances even without using any $f$.

We accomplish the second task by deriving a permutation $\vdf$ from a subexponentially
hard family of $\rSVL$ instances. The derived $\vdf$ is proven to be secure. Although it
suffices for the membership in the class $\VDF$, we also deduce a generic $\vdf$ from
the same $\rSVL$ instances but using an one-way injective function. 

\section{Related Work}\label{literature}
In this section, first we describe some well-known schemes that are qualified as $\vdf$s.

\subsubsection{Squaring over $\Z_p$} The pricing function by Dwork--Naor scheme~\cite{Dwork1992Price} asks a prover,
given a prime $p=3\; (\mathbf{mod}\; 4)$ and a quadratic residue $x$ modulo $p$,
to find a $y$ such that $y^2= x\;(\mathbf{mod}\;p)$. The prover
has no other choice other than using the identity $y= x^{\frac{(p+1)}{4}}\;(\mathbf{mod}\;p)$, 
but the verifier verifies the correctness using $y^2= x\;(\mathbf{mod}\;p)$. 
The drawback of this design is that the delay parameter $T=\O(\log p)$. Thus the difference between the evaluation 
and the verification may be made up by a prover with $\poly(T)$-processors by parallelizing the field multiplications.
Moreover, it is difficult to generate the public parameters of this $\vdf$ for sufficiently large $T$
as $\setup$ needs to sample a prime $p > 2^{\Omega(T)}$. 

\subsubsection{Injective Rational Maps}
In 2018, Dan et al.~\cite{Dan2018VDF} propose a $\vdf$ based on injective rational
maps of degree $T$, where the fastest possible inversion is to compute the polynomial
GCD of degree-$T$ polynomials. They conjecture that 
it achieves $(T^2,o(T))$ sequentiality using permutation polynomials as the candidate map.
However, it is a weak $\vdf$ as it needs $\mathcal{O}(T)$ processors to evaluate the output in time $T$.

\subsubsection{RSW Puzzle}\label{vdf}
Rivest, Shamir, and Wagner~\cite{Rivest1996Time} introduced the time-lock puzzle stating
that it needs at least $T$ number of sequential squaring to compute 
$y=g^{2^T}\mod{\Delta}$ when the factorization of $\Delta$ is unknown.
Therefore they proposed this encryption that can be decrypted 
only sequentially. Starting with $\Delta=pq$ such that $p,q$ are large primes,
the key $y$ is enumerated as $y=g^{2^T}\mod{\Delta}$. Then the verifier,
uses the value of $\phi(\Delta)$ to reduce the exponent to
$e=2^T\mod{\phi(\Delta)}$ and finds out $y= g^e\mod{\Delta}$.
On the contrary, without the knowledge of $\phi(\Delta)$, the only option available to the prover
is to raise $g$ to the power $2^T$ sequentially. 
As the verification stands upon a secret, the knowledge of $\phi(\Delta)$, 
it is not a $\vdf$ as verification should depend only on public parameters.

Pietrzak~\cite{Pietrzak2019Simple} and Wesolowski~\cite{Wesolowski2019Efficient} circumvent 
this issue independently. We describe both the $\vdf$s in the generic group $\mathbb{G}$
as these schemes can be instantiated over two different groups --
 the RSA group $\multgroup{\Delta}$ and the class group of imaginary quadratic
number field. Both the protocols use a common random oracle
$\RO_\mathbb{G}:\{0,1\}^*\rightarrow \mathbb{G}$ to map the input statement $x$ to the
generic group $\mathbb{G}$. We assume $g:=\RO_\mathbb{G}(x)$.

\begin{description}

\item [Pietrzak's $\vdf$]
It exploits the identity $z^ry=(g^rz)^{2^{T/2}}$
where $y=g^{2^T}$, $z=g^{2^{T/2}}$ and $r \in \Z_{2^\lambda}$ is chosen at random.
So the prover is asked to compute the output $y=g^{2^T}$ and the proof $\pi=\{u_1, u_2, \ldots, u_{\log{T}}\}$ such that 
$u_{i+1}=u_i^{r_i+{2^{T/2^i}}}$, $r_i = \RO(u_i, T/2^{i-1},v_i,u_i^{2^{T/2^{i-1}}})$ and
$v_i=u_i^{r_i\cdot{2^{T/2^i}}+2^T}$. The verifier computes the $v_i=u_i^{r_i\cdot{2^{T/2^i}}+2^T}$
and checks if $v_{\log{T}}=u_{\log T}^2$.
So the verifier performs $\sum_{1}^{\log{T}}\log{r_i}$ number of sequential
squaring. As $\RO$ samples $r_i$ uniformly from its range $\Z_{2^\lambda}$, we have
$\sum_{1}^{\log{T}}\log{r_i}=\O(\lambda \log{T}) $.
The effort to generate the proof $\pi$ is in
$\O(\sqrt{T} \log{T})$.

\item [Wesolowski's $\vdf$]
It asks the prover to compute an output $y=g^{2^T}$ and
a proof $\pi=g^{\lfloor2^T/\ell\rfloor}$, where 
$\ell=\RO_{\texttt{prime}}(\texttt{bin}(g)|||\texttt{bin}(y))$ is a $2\lambda$-bit prime. 
It needs $\O(T/\log {T})$ time to do the same.
The verifier checks if $y=\pi^\ell \cdot g^{(2^T \mod{\ell})}$. 
Hence the verification needs at most $2\log \ell=4 \lambda$ squaring.

\end{description}

\subsubsection{Isogenies over Super-Singular Curves}
Feo et al.~\cite{Feo2019Isogenie} presents two $\vdf$s based on isogenies over super-singular
elliptic curves. They start with five groups $\langle G_1,G_2,G_3,G_4,G_5\rangle$
of prime order $T$ with two non-degenerate bilinear pairing maps
$e_{12}: G_1 \times G_2 \rightarrow G_5$ and $e_{34}: G_3 \times G_4 \rightarrow G_5$.
Also there are two group isomorphisms
$\phi: G_1 \rightarrow G_3$ and $\overline{\phi}: G_4 \rightarrow G_2$. 
Given all the above descriptions as the public parameters along with a generator $P\in G_1$,
the prover needs to find $\overline{\phi}(Q)$, where $Q\in G_4$, using $T$ sequential steps.
The verifier checks if $e_{12}(P,\overline{\phi}(Q))=e_{34}(\phi(P),Q)$ in
$\poly(\log{T})$ time. It runs on super-singular curves over $\mathbb{F}_p$ and $\mathbb{F}_{p^2}$ 
as two candidate groups. 
While being inherently non-interactive, there are two drawbacks as mentioned by the authors themselves.
First, it requires a trusted setup, and second, the 
setup phase may turn out to be slower than the evaluation.

Mahmoody et al. recently rule out the possibility of having a $\vdf$ out 
of random oracles only~\cite{Mahmoody2020RO}.
 \begin{table*}[h]
 \caption{Comparison among the existing VDFs. $T$ is the targeted time bound, $\lambda$ is 
 the security parameter, $\Delta$ is the number of processors. All the quantities may be
 subjected to $\mathcal{O}$-notation, if needed.}
 \label{tab : VDF}
  \centering
  \begin{tabular}{|l@{\quad}|r@{\quad}|r@{\quad}|r@{\quad}|r@{\quad}|r@{\quad}|r@{\quad}}
     \hline
         VDF & \textsf{Eval} & \textsf{Eval} &  \textsf{Verify} & \textsf{Setup} & Proof   \\
 (by authors)&  Sequential   & Parallel      &                  &                &  size   \\
     \hline
     
     Dwork and Naor~\cite{Dwork1992Price}       & $T$   & $T^{2/3}$  &  $T^{2/3}$  & $T$ & $\textendash$ \\ 
     [0.3 em] \hline 
         
     Dan et al.~\cite{Dan2018VDF}         & $T^2$ & $>T-o(T)$  &  $\log{T}$  & $\log{T}$ & $\textendash$  \\
     [0.3 em] \hline
     
     Wesolowski~\cite{Wesolowski2019Efficient} & $(1+\frac{2}{\log{T}})T$   & $(1+\frac{2}{\Delta\log{T}})T$  &  $\lambda^{4}$  & $\lambda^{3}$ & $\lambda^{3}$ \\
     [0.3 em] \hline
     
     Pietrzak~\cite{Pietrzak2019Simple}        & $(1+\frac{2}{\sqrt{T}})T$   & $(1+\frac{2}{\Delta\sqrt{T}})T$  &  $\log{T}$  & $\lambda^{3}$ & $\log{T}$ \\
     [0.3 em] \hline 
     
     Feo et al.~\cite{Feo2019Isogenie}         & $T$   & $T$  &  $\lambda^4$  & $T\log{\lambda}$ & \textendash \\
     [0.3 em] \hline 
     
  \end{tabular}
 \end{table*}

\subsubsection{$\PPAD$-hardness of Cryptographic Protocols}
Now, we briefly mention few works in order to show the significance of the class $\PPAD$ 
in the context of cryptography.

Abbot, Kane and Valiant were the first to show that
virtual black-box obfuscation \cite{Valiant} can be used to generate hard instances of
\textsc{End-of-Line} ($\EOL$). Since, virtual black-box obfuscation is known only for
certain functions, Bitansky et al. consider indistinguishability obfuscation ($i\O$) to
show that quasi-polynomially hard $i\O$ and subexponentially hard one-way function
reduce to $\EOL$ via a new problem \textsc{Sink-of-Verifiable-Line}~\cite{BitanskyNash}. 
Building up further, Garg et al. derived
$\PPAD$-hardness from polynomially-hard $i\O$ or compact public-key functional encryption 
and one-way permutation~\cite{Garg16iO}.  
Relying on sub-exponentially hard injective one-way functions, Komargodski and Segev show that
quasi-polynomially hard private-key functional encryption implies $\PPAD$-hardness.

Hub\'{a}\v{c}ek and Yogev introduced a new total
search problem \textsc{End-of-Metered-Line} ($\EOML$) proving that it is hard to search local optima 
even over continuous domains. They also show that $\EOML$ belongs to a subclass, namely
continuous local search ($\CLS$), of $\PPAD$.

A striking result by Choudhuri et al. shows that relative to a random oracle (used in
Fiat--Shamir transformation), hardness in the class $\#\mathbf{P}$ implies
hardness in $\CLS$~\cite{ChoudhuriNash}. In particular, 
they derive a new verifiable procedure applying the
Fiat--Shamir transformation on the sumcheck protocol for \textsc{\#SAT} and reduce it  
to a new problem \textsc{Relaxed-Sink-of-Verifiable-Line} ($\rSVL$) in $\CLS$. 

A similar work by the same group of authors inspires our work~\cite{ChoudhuriRSW}. 
It reduces the problem of finding $g^{2^T}\mod{N}$ relative 
to a random oracle (used in Fiat--Shamir transformation) to $\rSVL$.
A typical property in Pietrzak's $\vdf$~\cite{Pietrzak2019Simple} 
called "proof-merging" is at the core of this reduction. 
Suppose, $\pi^T_{g \rightarrow y}=\{u_1, \ldots, u_{\log T}\}$ denotes 
the proof for $h=g^{2^T}$ in Pietrzak's $\vdf$. The property proof-merging is the
observation that given two proof $\pi^T_{g \rightarrow h}$ 
and $\pi^T_{h \rightarrow y}$, finding the proof $\pi^{2T}_{g \rightarrow y}$ reduces to
finding a proof $\pi^T_{u \rightarrow v}$ such that,
$u:=g^r \cdot h$,
$v:=h^r \cdot y$ and 
$r:=\RO(u_1,g,y,2T)$.
It works because the element $h$ must be present in the proof 
$\pi^{2T}_{g \rightarrow y}=\{u_1, \ldots,u_{\log T+1}\}$ 
as $u_1$. Therefore, the merged proof $\pi^T_{u \rightarrow
v}=\{u'_1,\ldots, u'_{\log T}\}$ is equivalent to the proof  
$\pi^{2T}_{g \rightarrow y}=\{u_1,u'_1, \ldots,u'_{\log T}\}$.

\section{Preliminaries}\label{preliminaries}

We start with the notations.

\subsection{Notations}
We denote the security parameter with $\lambda\in\mathbb{Z}^+$.
The term $\poly(\lambda)$ refers to some polynomial of $\lambda$, and
$\negl(\lambda)$ represents some function $\lambda^{-\omega(1)}$.
If any randomized algorithm $\mathcal{A}$ outputs $y$ on an input $x$, 
we write $y\xleftarrow{R}\mathcal{A}(x)$. By $x\xleftarrow{\$}\mathcal{X}$,
we mean that $x$ is sampled uniformly at random from $\mathcal{X}$. For a string $x$, 
$|x|$ denotes the bit-length of $x$, whereas for any set $\mathcal{X}$, $|\mathcal{X}|$ denotes 
the cardinality of the set $\mathcal{X}$. If $x$ is a string then $x[i \ldots j]$ denotes the substring 
starting from the literal $x[i]$ ending at the literal $x[j]$. 
We consider an algorithm $\adv$ as efficient if it runs in 
probabilistic polynomial time (PPT). 
 
\subsection{Verifiable Delay Function}\label{VDF}
We borrow this formalization from~\cite{Dan2018VDF}.

\begin{definition}\normalfont{ \bf (Verifiable Delay Function).}
A verifiable delay function from domain $\X$ to range $\Y$ is a tuple of algorithms 
$(\setup, \eval, \verify)$ defined as follows,
\begin{itemize}[label=\textbullet]
\item \textsf{Setup}$(1^\lambda, T) \rightarrow \pp$
 is a randomized algorithm that takes as input a security parameter $\lambda$ 
 and a targeted time bound $T$, and produces the public parameters 
 $\pp$. We require \textsf{Setup} to run in $\poly(\lambda,\log{T})$ time.
 
 \item \textsf{Eval}$(\pp, x) \rightarrow (y, \pi)$ takes an input 
 $x\in\mathcal{X}$, and produces an output $y\in\mathcal{Y}$ and a (possibly empty) 
 proof $\pi$. \textsf{Eval} may use random bits to generate the proof 
 $\pi$. For all $\pp$ generated 
 by $\textsf{Setup}(\lambda, T)$ and all $x\in\mathcal{X}$, the algorithm 
 \textsf{Eval}$(\pp, x)$ must run in time $T$.
 
 \item \textsf{Verify}$(\pp, x, y, \pi) \rightarrow \{0, 1\}$ is a 
 deterministic algorithm that takes an input $x\in\mathcal{X}$, an output $y\in\mathcal{Y}$,
 and a proof $\pi$ (if any), and either accepts ($1)$ or rejects ($0)$. 
 The algorithm must run in $\poly(\lambda,\log{T})$ time.
\end{itemize}
\end{definition}

Before we proceed to the security of $\vdf$s we need the precise model of parallel
adversaries \cite{Dan2018VDF}. 
\begin{definition}\normalfont{(\bf Parallel Adversary)}\label{paradv} 
A parallel adversary $\adv=(\adv_0,\adv_1)$ is a pair of non-uniform 
randomized algorithms $\adv_0$ with total running time $\poly(\lambda,T)$, 
and $\adv_1$ which runs in parallel time $\sigma(T)<T-o(T)$ on at 
most $\poly(\lambda,T)$ number of processors.
\end{definition}
Here, $\adv_0$ is a preprocessing algorithm that precomputes some
$\st$ based only on the public parameters, and $\adv_1$ exploits
this additional knowledge to solve in parallel running time $\sigma$ on 
$\poly(\lambda,T)$ processors.

The three desirable properties of a $\vdf$ are now introduced.

\begin{definition}\normalfont{(\bf Correctness)}\label{def: Correctness} 
A $\vdf$ is correct with some error probability $\varepsilon$,
if for all $\lambda, T$, parameters $\pp$, 
and $x\in\mathcal{X}$, we have
\[
\Pr\left[
\begin{array}{l}
\textsf{Verify}(\pp,x,y,\pi)=1
\end{array}
\Biggm| \begin{array}{l}
\pp\leftarrow\textsf{Setup}(1^\lambda,T)\\
x\xleftarrow{\$} \mathcal{X}\\
(y,\pi)\leftarrow\textsf{Eval}(\pp,x)
\end{array}
\right]
\ge 1 - \negl(\lambda).
\]
\end{definition}

\begin{definition}\normalfont{\bf(Soundness)}\label{def: Soundness} 
A $\vdf$ is computationally sound if for all non-uniform algorithms $\adv$ 
that run in time $\mbox{poly}(T,\lambda)$,
we have
\[
\Pr\left[
\begin{array}{l}
y\ne\textsf{Eval}(\pp,x)\\
\textsf{Verify}(\pp,x,y,\pi)=1
\end{array}
\Biggm| \begin{array}{l}
\pp\leftarrow\textsf{Setup}(1^\lambda,T)\\
(x,y,\pi)\leftarrow\mathcal{A}(1^\lambda,T,\pp)
\end{array}
\right] \le \negl(\lambda).
\]
\end{definition}

Further, a $\vdf$ is called statistically sound when all adversaries 
(even computationally unbounded) have at most $\negl(\lambda)$ advantage.
Even further, it is called perfectly sound if we want this probability to be $0$ 
against all adversaries. Hence, perfect soundness implies statistical soundness which
implies computational soundness but not the reverses. 

\begin{definition}\normalfont{\bf (Sequentiality)}\label{def: Sequentiality}
A $\vdf$ is $(\Delta,\sigma)$-sequential if there exists no
pair of randomized algorithms $\adv_0$ with total running time
$\mbox{poly}(T,\lambda)$ and $\adv_1$ which runs
in parallel time $\sigma$ on at most $\Delta$ processors, such that
\[
\Pr\left[
\begin{array}{l}
y=\textsf{Eval}(\pp,x)
\end{array}
\Biggm| \begin{array}{l}
\pp\leftarrow\textsf{Setup}(1^\lambda,T)\\
\st\leftarrow\mathcal{A}_0(1^\lambda,T,\pp)\\
x\xleftarrow{\$}\mathcal{X}\\
y\leftarrow\mathcal{A}_1(\st,x)
\end{array}
\right]
\le \negl(\lambda).
\]
\end{definition}


\begin{definition}{\normalfont \textbf{(Permutation $\vdf$).}}
Permutation $\vdf$s are the $\vdf$s with $\X=\Y$
where $\X$ and $\Y$ denote the input and output domains respectively.
\end{definition}

We reiterate an important remark from~\cite{Dan2018VDF} but as a lemma.
\begin{lemma}{\normalfont ($T \in \mathsf{SUBEXP(\lambda)}).$}\label{subexp}
If $T > 2^{o(\lambda)}$ then there exists an adversary that breaks the sequentiality of
the $\vdf$ with non-negligible advantage.
\end{lemma}
\begin{proof}
$\adv$ observes that the algorithm $\verify$ is efficient.
So given a statement $x \in \X$,  $\adv$ chooses an arbitrary $y \in \Y$ as the output
without running $\eval(x,\pp,T)$. Now, $\adv$ finds the proof $\pi$ by a brute-force
search in the entire solution space with its $\poly(T)$ number of processors.
In each of its processors, $\adv$ checks if $\verify(x,\pp,T,y,\pi_i)=1$ with different
$\pi_i$. The advantage of $\adv$ is $\poly(T)/2^{\Omega(\lambda)} \ge \negl(\lambda)$ as
$T > 2^{o(\lambda)}$. 

\end{proof}
So we need $T \le 2^{o(\lambda)}$ to restrict the advantage of $\adv$ 
upto $2^{o(\lambda)}/2^{\Omega(\lambda)}=2^{-\Omega(\lambda)}$.

\subsection{Search Problems}
In this section, we review the basics of search problems
from~\cite{ChoudhuriNash,BitanskyNash,PPAD94,TFNP91}. 

Suppose, $R \subseteq \{0,1\}^* \times \{0,1\}^*$ 
is a relation such that, for all $(x,y) \in R$,
\begin{enumerate}[label=\roman*]
\item $R$ is polynomially-balanced i.e., $|y|\le \poly(|x|)$. 
\item $R$ is efficiently-recognizable. 
\end{enumerate}

A search problem $(\L,R)$ is defined by a set of instances $\L \subseteq \{0,1\}^*$ and
a relation $R$. In particular, given a $x \in \L$, the search problem $(\L,R)$ is to
find an $y$ if there exists an $(x,y)\in R$, otherwise say "no". The set of all search problems 
is called as functional-$\NP$ or $\FNP$. For example, the search version 
of 3-$\mathsf{SAT}$ belongs to $\FNP$. 

The relation $R$ is called total if, for every $x$, there is always a $y$ such that
$(x,y)\in R$. A search problem $(\L,R)$ is called total when $R$ is total. It means 
total search problems always have solutions e.g., $\mathsf{FACTORING}$. The set 
of all total search problems is called as total-$\FNP$ or $\TFNP$.

A notable subclass of $\TFNP$ is called as $\PPAD$ which stands for "polynomial 
parity argument in a directed graph". It is defined as the set of all problems that are
polynomial-time reducible in \textsc{End-of-Line} problem~\cite{PPAD94}.

\begin{definition}{\normalfont \textbf{(\textsc{End-of-Line} problem $\EOL$).}} 
An \textsc{End-of-Line} instance $(\S,\P)$ consists of a pair of circuits
$\S,\P: \{0,1\}^n \rightarrow\{0,1\}^n$ such that $\P(0^n)=0^n$ and
$\S(0^n) \ne 0^n$. The goal is to find a vertex $v\in\{0,1\}^n$ 
such that $\P(\S(v))\ne v$ or $\S(P(v))\ne v \ne 0^n$.
\end{definition}

Thus, $\EOL$ deals with a directed graph over the vertices $\{0,1\}^n$ and the edges of
the form $(u,v)$ if and only if $\S(u)=v$ and $\P(v)=u$. Here, $\S$ and $\P$
represent the successor and predecessor functions for this directed graph. The in-degree
and out-degree of every vertex in this graph is at most one except that in-degree of
$0^n$ is 0. the goal is to find a vertex $v$, other than $0^n$, which is either 
a source (in-degree is 0) or a sink (out-degree is 0). Such a vertex always exists by
the parity argument for a graph - the number of odd degree vertex in a graph is even
Therefore, $\EOL$ is in $\TFNP$ and by definition in $\PPAD$.

A subclass of $\PPAD$, namely continuous local search $\CLS$ is 
the class of problems that are polynomial-time reducible to the problem
\textsc{Continuous-Local-Optimum}~\cite{DaskalakisCLS}.
An interesting but not known to be complete problem in 
$\CLS$ is \textsc{End-of-Metered-Line}~\cite{HubacekCLS}.  

\begin{definition}{\normalfont \textbf{(\textsc{End-of-Metered-Line} problem $\EOML$).}} 
An \textsc{End-of-Metered-Line} instance $(\S,\P,\M)$ consists of circuits
$\S,\P: \{0,1\}^n \rightarrow\{0,1\}^n$ and $\M:\{0,1\}^n \rightarrow \{0,\ldots,2^n-1\}$
such that $\P(0^n)=0^n\ne \S(0^n)$ and
$\M(0^n) =1$. The goal is to find a vertex $v\in\{0,1\}^n$ 
satisfying one of the following,

\begin{enumerate}[label=\roman*]
\item \textbf{End of Line:} either $\P(\S(v))\ne v$ or $\S(\P(v))\ne v \ne 0^n$.
\item \textbf{False source:} $v\ne 0^n$ and $\M(v)=1$.
\item \textbf{Miscount:} either $\M(v)>0 $ and $\M(\S(v))-\M(v) \ne 1$ or $\M(v) >1$
and $\M(v)-\M(\P(v))\ne 1$.
\end{enumerate}
\end{definition}

Clearly, $\EOML$ reduces to $\EOL$, but the "odometer" circuit $\M$ makes $\EOML$
easier than $\EOL$. The circuit $\M$ outputs the number of steps required to reach $v$
from the source. Observe that any vertex, for which $\M$ contradicts its correct
behaviour, solves the problem. Thus, there exists a solution for every $\EOML$ instance.
Hence, $\EOML \in \TFNP$.

The problem \textsc{Sink-of-Verifiable-Line}, introduced by Valiant et. al and further, 
developed in~\cite{BitanskyNash}, is defined as follows,

\begin{definition}{\normalfont \textbf{(\textsc{Sink-of-Verifiable-Line} problem $\SVL$).}} 
An \textsc{Sink-of-Verifiable-Line} instance $(\S,\V,T,v_0)$ consists of 
$T\in\{1,\ldots,2^n\}$, $v_0 \in \{0,1\}^n$, and two circuits
$\S: \{0,1\}^n \rightarrow\{0,1\}^n$ and $\V:\{0,1\}^n \times
\{1,\ldots,T\}\rightarrow \{0,1\}$ with the guarantee
that for every $v \in \{0,1\}^n$ and
$i\in\{1,\ldots,T\}$, it holds that $\V(v,i)=1$ if and only if $v=\S^i(v_0)$. 
The goal is to find a vertex $v\in\{0,1\}^n$ such that $\V(v,T)=1$ (i.e., the sink). 
\end{definition}
 
Similar to, $\EOL$ and $\EOML$ the circuit $\S$ implements a successor function for the
directed graph. However, in $\SVL$ the graph is a single line with the source $v_0$. The
circuit $\V$ allows to test if a vertex $v$ is at a distance of $i$ on the line from
$v_0$. The goal is to find the vertex at distance $T$ from $v_0$.
Given an arbitrary instance $(\S,\V,T,v_0)$, 
we do not know how to efficiently check if $\V$ behaves correctly. 
Therefore, every instance of $\SVL$ may not be valid and may not have solutions. 
So $\SVL \notin \TFNP$. 

Although, in \cite{BitanskyNash}, $\SVL$ was defined with the fixed source $v_0=0^n$, 
it is equivalent to $\SVL$ with arbitrary source. The instance $(\S,\V,T,0^n)$ reduces
to the instance $(\S,\V,T,v_0)$ considering $v_0=0^n$. On the other hand, the instance 
$(\S,\V,T,v_0)$ reduces to $(\S',\V',T,0^n)$ when we define $\S'(v):=\S(v \oplus v_0)$ and 
$\V'(v,i):=\V(v \oplus v_0,i)$. Most importantly, this reduction works for any search
problem in the context of $\TFNP$ that considers fixed vertex in its input (like $\EOL$ with $0^n$).

In order to solve the issue $\SVL \notin \TFNP$, Choudhuri et al. introduced a relaxed
version of it, namely \textsc{Relaxed-Sink-of-Verifiable-Line}. In this version, the
circuit $\V$ allows a few vertices off the main line starting at $v_0$. So, these
off-the line vertices also act as solutions. Hence, 

\begin{definition}{\normalfont \textbf{(\textsc{Relaxed-Sink-of-Verifiable-Line} problem $\rSVL$).}} 
An \textsc{Relaxed-Sink-of-Verifiable-Line} instance $(\S,\V,T,v_0)$ consists of 
$T\in\{1,\ldots,2^n\}$, $v_0 \in \{0,1\}^n$, and two circuits
$\S: \{0,1\}^n \rightarrow\{0,1\}^n$ and $\V:\{0,1\}^n \times
\{1,\ldots,T\}\rightarrow \{0,1\}$ with the guarantee
that for every $v \in \{0,1\}^n$ and
$i\in\{1,\ldots,T\}$, it holds that $\V(v,i)=1$ if and only if $v=\S^i(v_0)$. 
The goal is to find:
\begin{enumerate} 
\item \textbf{The sink:} a vertex $v \in \{0,1\}^n$ such that $\V(v,T)=1$ or
\item \textbf{False positive:} a pair $(v,i) \in \{0,1\}^n \times \{1,\ldots, 2^n\}$ 
such that $v \ne \S^i(v_0)$
and $\V(v,i)=1$. 
\end{enumerate} 
\end{definition}

Lemma. 10 in cf.~\cite{ChoudhuriNash} shows that \textsc{Relaxed-Sink-of-Verifiable-Line}
is many-one reducible to \textsc{End-of-Metered-Line}. Since, $\EOML \in \CLS$, 
thus $\rSVL \in \CLS \subset \PPAD$. The off the line vertices guarantee the existence of the
solutions for any arbitrary instances.

\subsection{Interactive Proof System}
Goldwasser et al. were the first to show that the interactions between the prover and
randomized verifier recognizes class of languages larger than $\NP$~\cite{Goldwasser85Knowledge}. 
They named the class as $\IP$ and the model of
interactions as the interactive proof system. Babai and Moran introduced the same notion of interactions in
the name of Arthur-Merlin games however with a restriction on the verifiers' side \cite{Babai88AM}. 
Later, Goldwasser and Sipser proved that both the models are equivalent
\cite{Goldwasser86IPAM}. Two important works in this context that motivate our present
study are by the Shamir showing that $\IP=\PSPACE$ \cite{ShamirIP} and by the Goldwasser
et al. proving that $\PSPACE = \mathbf{ZK}$, the set of all zero-knowledge protocols.
We summarize the interactive proof system from \cite{Goldwasser88ZK}.

An interactive proof system $\ips$ consists of a pair of Turing machines ($\TM$), $\prv$ and $\vrf$,
with common alphabet $\Sigma=\{0,1\}$. $\prv$ and $\vrf$ each have distinguished initial and
quiescent states. $\vrf$ has distinguished halting states out of which there is no
transitions. $\prv$ and $\vrf$ operates on various one-way infinite tapes,
\begin{enumerate}[label=\roman*.]
\item $\prv$ and $\vrf$ have a common read-only input tape.
\item $\prv$ and $\vrf$ each have a private random tape and a private work tape.
\item $\prv$ and $\vrf$ have a common communication tape.
\item $\vrf$ is polynomially time-bounded. This means $\vrf$ halts on input $x$ in time
$\poly(|x|)$. $\vrf$ is in quiescent state when $\prv$ is running.
\item $\prv$ is computationally unbounded but runs in finite time. This means $\prv$ 
may compute any arbitrary function $\{0,1\}^*\rightarrow \{0,1\}^*$ on input $x$ 
in time $f(|x|)$. Feldman proved that ``the optimum prover lives in $\PSPACE$"
\footnote{We could not find a valid citation.}.
\item The length of the messages written by $\prv$ into the common communication tape is
bounded by $\poly(|x|)$. Since $\vrf$ runs in $\poly(|x|)$ time, it can not write
messages longer than $\poly(|x|)$.
\end{enumerate}

Execution begins with $\prv$ in its quiescent state and $\vrf$ in its start state.
$\vrf$'s entering its quiescent state arouses $\prv$, causing it to transition to its
start state. Likewise, $\prv$'s entering its quiescent state causes $\vrf$ to
transition to its start state. Execution terminate when $\vrf$ enters in its halting
states. Thus $\ips(x)=1$ denotes $\vrf$ accepts $x$ and $\ips(x)=0$ denotes $\vrf$
rejects $x$.

\begin{definition}{\normalfont \textbf{(Interactive Proof System $\ips$ ).}}
$\ips$ is an interactive proof system for the language $\L\subseteq \{0,1\}^*$ if 

\begin{description}
 \item \noindent {\normalfont (Correctness).} $(x \in \L) \implies \Pr[\ips(x))=1] \ge 1-\negl(|x|)$.
 \item \noindent {\normalfont (Soundness).}  $(x\notin\L) \implies
\forall\prv',\Pr[(\prv'\leftrightarrow\vrf)(x))=1] < \negl(|x|)$.
\end{description}
\end{definition}

The class of polynomial-time interactive proofs $\IP$ is defined as the class of the languages
that have $\ips$ such that $\prv$ and $\vrf$ talk for at most $\poly(n)$-rounds. 
Thus $$\IP=\{\L \mid \L \text{ has a } \poly(n)\text{-round } \ips\}.$$
Alternatively and more specifically,

\begin{definition}{\normalfont \textbf{(The Class IP).}}
$$\IP=\bigcup_{k \in\poly(n)}\IP[k].$$
For every $k$, $\IP[k]$ is the set of languages $\L$ such that there exist a probabilistic polynomial time $\TM$
$\vrf$ that can have a $k$-round interaction with a prover $\prv : \{0,1\}^*\rightarrow \{0,1\}^*$ 
having these two following properties
\begin{description}
 \item \noindent {\normalfont (Correctness).} $(x \in \L) \implies \Pr[\ips(x))=1] \ge 1-\negl(|x|)$.
 \item \noindent {\normalfont (Soundness).}  $(x\notin\L) \implies
\forall\prv',\Pr[(\prv'\leftrightarrow\vrf)(x))=1] < \negl(|x|)$.
\end{description}
\end{definition}

%

\section{Fiat--Shamir Transformation}
Any interactive protocol $\ips$ can be transformed into a non-interactive
protocol if the messages from the verifier $\vrf$ are replaced with the response of
a random oracle $\RO$. This is known as Fiat--Shamir transformation (FS)~\cite{FS86}. In
particular, the $i$-th message from $\vrf$ is computed as
$y_i:=\RO(x,x_1,y_1,\ldots,x_i, y_{i-1})$ where $x_i$ denotes the $i$-th
response of $\prv$. When $\RO$ is specified in the public parameters of a $k$-round
protocol, the transcript $x,x_1,y_1,\ldots,x_k, y_{k-1}$ can be verified publicly.
Thus, relative to a random oracle $\RO$, a $k$-round interactive proof protocol 
$\ips$ can be transformed into a two-round non-interactive argument 
$(\prv_{FS} \leftrightarrow \vrf_{FS})$ where $\prv_{FS}$ sends the entire transcript 
$x,x_1,y_1,\ldots,x_k, y_k$ to $\vrf_{FS}$ in a single round. Under the assumption that $\RO$ is
one-way and collision-resistant, $\vrf_{FS}$ accepts $x\in \L$ in the next round if and only if 
$\vrf$ accepts. Here we summarize two claims on Fiat--Shamir transformation stated in~\cite{Ephraim20VDF}.

\begin{lemma}\label{FS}
If there exists an adversary $\adv$ who breaks the soundness of the non-interactive
protocol $(\prv_{FS} \leftrightarrow \vrf_{FS})$ with the probability $p$ using
$q$ queries to a random oracle then there exists another adversary $\adv'$ who breaks
the soundness of the $k$-round interactive protocol $\ips$ 
with the probability $p/q^k$.
\end{lemma}

\begin{proof}
See~\cite{Goldreich96} for details. 
\end{proof}

\begin{lemma}\label{soundFS}
Against all non-uniform probabilistic polynomial-time adversaries,  
if a $k$-round interactive protocol $(\prv \leftrightarrow \vrf)$
achieves $\negl(|x|^k)$-soundness then the non-interactive protocol 
$(\prv_{FS} \leftrightarrow \vrf_{FS})$ has $\negl(|x|)$-soundness.
\end{lemma}

\begin{proof}
Since, all the adversaries run in probabilistic polynomial time, the number of queries
$q$ to the random oracle must be upper-bounded by $\poly(|x|)$.
Putting $q =|x|^c $ for any $c \in \Z^+$ in lemma.~\ref{FS}, it follows the claim. 

\end{proof}

\section{$\vdf$ Characterization}
 
In this section, we investigate the possibility to model $\vdf$s as a language in order
to define its hardness. It seems that there are two hurdles,
\begin{description}
\item [Eliminating Fiat--Shamir]
The prover $\prv$ in Def.~\ref{VDF}, generates the proof $\pi:=f(x,y,T,\RO(x,y,T))$ using Fiat--Shamir
transformation where $y:=\eval(x,\pp,T)$. Unless Fiat--Shamir is eliminated from $\vdf$, its
hardness remains relative to the random oracle $\RO$. Sect.~\ref{ivdf} resolves this
issue.

\item  [Modelling Parallel Adversary] How to model the parallel adversary $\adv$ (Def. \ref{paradv}) 
in terms of computational complexity theory? We model $\adv$ as a special variant of
Turing machines described in Def.~\ref{PTM}.
\end{description}

We address the first issue now.

\subsection{Interactive $\vdf$s}\label{ivdf}

We introduce the interactive $\vdf$s in order to eliminate the Fiat--Shamir.
In the interactive version of a $\vdf$, the $\vrf$ replaces the
randomness of Fiat--Shamir heuristic. 
In particular, a non-interactive $\vdf$ with the Fiat--Shamir transcript 
$\langle x,x_1,y_1,\ldots,x_k, y_k \rangle$ can be translated into an equivalent $k$-round 
interactive $\vdf$ allowing $\vrf$ to choose $y_i$s in each round.

%

\begin{definition}\normalfont{ \bf (Interactive Verifiable Delay Function).} 
An interactive verifiable delay function is a tuple $(\setup, \eval,\mathsf{Open} ,\verify)$ 
that implements a function $\X\rightarrow\Y$ as follows,
\begin{itemize}[label=\textbullet]
\item $\setup(1^\lambda, T) \rightarrow \pp$
 is a randomized algorithm that takes as input a security parameter $\lambda$ 
 and a delay parameter $T$, and produces the public parameters 
 $\pp$ in $\poly(\lambda,\log{T})$ time.
 
 \item $\eval(\pp, x) \rightarrow y$ takes an input 
 $x\in\mathcal{X}$, and produces an output $y\in\mathcal{Y}$. 
 For all $\pp$ generated by $\textsf{Setup}(\lambda, T)$ 
 and all $x\in\mathcal{X}$, the algorithm 
 \textsf{Eval}$(\pp, x)$ must run in time $T$.

 \item $\mathsf{Open}(x,y,\pp,T,t)\rightarrow\pi$ 
 takes the challenge $t$ chosen by $\vrf$ and recursively computes 
 a proof $\pi$ in $k \in \poly(\lambda, \log T)$-rounds of 
 interaction with $\vrf$. In general, for some $k \in \poly(\lambda,\log T)$, 
 $\pi=\{\pi_1, \ldots, \pi_k\}$ can be computed as 
 $\pi_{i+1}:=\mathsf{Open}(x_i,y_i,\pp,T,t_i)$ where 
 $x_i$ and $y_i$ depend on $\pi_i$. Observing $(x_i,y_i,\pi_i)$ 
 in the $i$-th round, $\vrf$ chooses the challenge $t_i$ for 
 the $(i+1)$-th round. Hence, $\vrf$ can efficiently collect all the proofs 
 $\pi=\{\pi_1, \ldots, \pi_k\}$ in $\poly(\lambda,\log T)$-rounds of interactions.
 $\mathsf{Open}$ does not exist for the $\vdf$s that need no proof (e.g., 
 \cite{Feo2019Isogenie}). 
 \item $\verify(\pp, x, y, \pi) \rightarrow \{0, 1\}$ is a 
 deterministic algorithm that takes an input $x\in\mathcal{X}$, an output $y\in\mathcal{Y}$,
 and the proof vector $\pi$ (if any), and either accepts ($1)$ or rejects ($0)$. 
 The algorithm must run in $\poly(\lambda,\log{T})$ time.
\end{itemize}
\end{definition}

All the three security properties remain same for the interactive $\vdf$. 
Sequentiality is preserved by the fact that
$\mathsf{Open}$ runs after the computation of $y:=\eval(x,\pp,T)$. For soundness, we rely
on lemma.~\ref{soundFS}. The correctness of interactive $\vdf$s implies the correctness
of the non-interactive version as the randomness that determines the proof is not in the
control of $\prv$. Therefore, an honest prover always convinces $\vrf$. 

%
%
%

Although the interactive $\vdf$s do not make much sense as publicly verifiable proofs in
decentralized distributed networks, it allow us to analyze its hardness irrespective of
any random oracle.  

In order to model parallel adversary, we consider a well-known variant of Turing machine that suits the
context of parallelism. We describe the variant namely parallel Turing machine as
briefly as possible from (Sect. 2 in cf.~\cite{Worsch1993Parallel})%
 \subsection{Parallel Turing Machine}
 Intuitively, a parallel Turing machine has multiple control units (\textsf{CU})
 (working collaboratively) with a single head associated with each of them
 working on a common read-only input tape \cite{Worsch1993Parallel}. 
 and a common read-write work tape.
 \begin{definition}{\normalfont \textbf{(Parallel Turing Machine).}}\label{PTM}
  a parallel Turing machine is a tuple 
  $\PTM = \langle Q, \Gamma,\Sigma, q_0, F, \delta \rangle$ where
\begin{enumerate}
 \item $Q$ is the finite and nonempty set of states.
 \item $\Gamma$ is the finite and non-empty set of tape alphabet symbols including the input alphabet $\Sigma$.
 \item $q_0\in Q$ is the initial state.
 \item $F \subseteq Q$ is the set of halting states.
 \item $\delta : 2^Q \times \Gamma\rightarrow 2^{Q\times D} \times \Gamma$ 
 where $D=\{-1,0,+1\}$ is the set of directions along the tape.
\end{enumerate}
 \end{definition}
 
 A configuration of a $\PTM$ is a pair $c=(p,b)$ of mappings
 $p:\Z^+\rightarrow 2^Q$ and $b:\Z^+\rightarrow \Gamma$.
 The mapping $p(i)$ denotes the set of states of the \textsf{CU}s
 currently pointing to the $i$-th cell in the input tape and $b(i)$ is the 
 symbol written on it. So it is impossible for two different \textsf{CU}s 
 pointing to the same cell $i$ while staying at the same state simultaneously.
 During transitions $c'=(M'_i,b'(i))=\delta(c)=\delta(p(i),b(i))$, the set of \textsf{CU}s
 may be replaced by a new set of \textsf{CU}s $M'_i \subseteq Q \times D$. 
 The $p'(i)$ in the configuration $c'$ is defined as 
 $p'(i)=\{q \mid (q,+1)\in M'_{i-1} \lor (q,0)\in M'_{i} \lor (q,-1)\in M'_{i+1}\}$. 
 
 Without loss of generality, the cell $1$ is observed in order to find 
 the halting condition of $\PTM$. We say that a $\PTM$ halts on a 
 string if and only if $p(1)\subseteq F$ after some finite time. 
 The notion of decidability by a $\PTM$ 
 is exactly same as in $\TM$. We denote $\PTM(s,t,h)$
 as the family of all languages for which there is a $\PTM$ recognizing them 
 using space $s$, time $t$ and $h$ processors. Thus languages decidable by a
 $\TM$ is basically decidable by a $\PTM(s,t,1)$. Assuming $\TM(s,t)$ is the set of
 languages recognized by a $\TM$ in space $s$ and time $t$, we mention Theorem 15 from
 (cf. \cite{Worsch1993Parallel}) without the proof.


We observe that the parallel adversary $\adv$ defined in Def. \ref{paradv} is
essentially a $\PTM$ having $\poly(\lambda,T)$ processors running on $\poly(\lambda,T)$
space in time $\sigma(T)$. We will refer such a $\PTM$ with $\poly(\lambda,T)$-$\PTM$
(w.l.o.g.) in our subsequent discussions.

\subsection{$\vdf$ As A Language} 

Now  we characterize $\vdf$s in terms of computational complexity theory. We
observe that, much like $\ips$, $\vdf$s are also proof system for the languages,
\[
\L=\left\{(x,y,T)
\begin{array}{l}
\end{array}
\Biggm| \begin{array}{l}
\pp\leftarrow\setup(1^\lambda,T)\\
x \in  \{0,1\}^\lambda\\
y\leftarrow\eval(\pp,x)
\end{array}
\right\}.
\]
$\prv$ tries to convince $\vrf$ that the tuple $(x,y,T)\in \L$ in
polynomially many rounds of interactions. 
In fact, Pietrzak represents his $\vdf$ using such a language (Sect. 4.2 in cf.
\cite{Pietrzak2019Simple}) where it needs $\log T$ (i.e., $\poly(\lambda))$
rounds of interaction. However, by design, the $\vdf$ is non-interactive. It uses
Fiat--Shamir transformation.

Thus, a $\vdf$ closely resembles an $\ips$ except on the fact
that it stands sequential (see Def.~\ref{def: Sequentiality}) even against 
an adversary (including $\prv$) possessing subexponential parallelism. 
Notice that a $\poly(\lambda,T)$-$\PTM$ (see Def. \ref{PTM})
precisely models the parallel adversary described in Def. \ref{paradv}. 
In case of interactive proof systems, we never talk about the running time of $\prv$
except its finiteness. On the contrary, $\prv$ of a $\vdf$ must run for at least
$T$ time in order to satisfy its sequentiality. Hence, we define $\vdf$ as follows,

\begin{definition}{\normalfont \textbf{(Verifiable Delay Function $\vdfs$).}}
For every $\lambda\in \Z^+$, $T \in 2^{o(\lambda)}$ and for all $s= (x,y,T) \in
\{0,1\}^{2\lambda +\lceil\log T\rceil}$,
$\vdfs$ is a verifiable delay function for a language $\L\subseteq\{0,1\}^*$ if
\begin{description}
 \item \noindent {\normalfont (Correctness).} $(s \in \L) \implies \Pr[\vdfs(s))=1]\ge 1-\negl(\lambda)$.
 \item \noindent {\normalfont (Soundness).} $(s \notin\L)\implies\forall\adv,\Pr[\advs(s))=1] \le \negl(\lambda)$.
 \item \noindent {\normalfont (Sequentiality).} $(s\in\L)\implies\forall\adf,\Pr[\adfs(s))=1] \le \negl(\lambda)$.
\end{description}
where, 
\begin{enumerate}[label=\roman*.]
\item $\prv : \{0,1\}^*\rightarrow \{0,1\}^*$ is a $\TM$ that runs in time $\ge T$, 
\item $\adv : \{0,1\}^* \rightarrow \{0,1\}^*$ is a $\TM$ that runs in 
time $\poly(\lambda,T)$,
\item $\adf$ is a {\normalfont $\poly(\lambda,T)$}-$\PTM$ 
(see Def. \ref{PTM}) that runs in time $<T$.
\end{enumerate}
\end{definition}


Further, we define the class of all verifiable delay functions as, 
\begin{definition}{\normalfont \textbf{(The Class $\VDF)$.}}\label{VDF}
$$\VDF=\bigcup_{k \in \poly(\lambda)}  \VDF[k].$$
For every $k\in\Z^+$, $\VDF[k]$ is the set of languages $\L$ 
such that there exists a probabilistic polynomial-time 
$\TM$ $\vrf$ that can have a {\normalfont $k$}-round interaction with
\begin{enumerate}[label=\roman*.]
\item $\prv : \{0,1\}^* \rightarrow \{0,1\}^*$ is a $\TM$ that runs in time $\ge T$, 
\item $\adv : \{0,1\}^* \rightarrow \{0,1\}^*$ is a $\TM$ that runs in time $\poly(\lambda,T)$,
\item $\adf$ is a {\normalfont $\poly(\lambda,T)$}-$\PTM$ 
(see Def. \ref{PTM}) that runs in time $<T$.
\end{enumerate}
satisfying these three following properties,

\begin{description}
 \item \noindent {\normalfont (Correctness).} $(s \in \L) \implies \Pr[\vdfs(s))=1]\ge 1-\negl(\lambda)$.
 \item \noindent {\normalfont (Soundness).} $(s \notin\L)\implies\forall\adv,\Pr[\advs(s))=1] \le \negl(\lambda)$.
 \item \noindent {\normalfont (Sequentiality).} $(s\in\L)\implies\forall\adf,\Pr[\adfs(s))=1] \le \negl(\lambda)$.
\end{description}
\end{definition}

\section{$\VDF$-completeness of $\rSVL$}
In this section, we show that $\rSVL$ is a complete problem the class $\VDF$. 
\begin{theorem}{\normalfont \textbf{($\rSVL$ is $\VDF$-complete).}}\label{thm:complete}
$\rSVL$ is a complete problem for the class $\VDF$.  
\end{theorem}

\begin{proof}
By theorem.~\ref{thm:member} $\rSVL$ belongs to the class $\VDF$.
By theorem.~\ref{thm:hardness}, $\rSVL$ is a hard problem for the class $\VDF$.
Hence, $\rSVL$ is a complete problem for the class $\VDF$.  
\end{proof}

\begin{theorem}{\normalfont \textbf{($\rSVL\in \VDF$).}}\label{thm:member}
For the parameters $\lambda\in \Z^+$ and $T=T(\lambda) \in 2^{o(\lambda)}$, let 
$\S : \{0,1\}^\lambda \rightarrow \{0,1\}^\lambda$, $\V : \{0,1\}^\lambda \times \{1,
\ldots, T\}\rightarrow \{0,1\}$ 
and $x\in\{0,1\}^\lambda$. If there exists a family of $\rSVL$ instances 
$\{\S,\V,v,T\}_{v\in\{0,1\}^\lambda}$ such that each 
instance allows at most polynomially many (i.e., $\poly(\lambda)$) false positive vertices  
then there exists a permutation $\vdf$s with $\X=\Y=\{0,1\}^\lambda$ 
and the delay parameter $T$.
\end{theorem}

\begin{proof}
Given any $\lambda\in \Z^+$ and any $T \in 2^{o(\lambda)}$,
we derive a permutation $\vdf$ from a sub-family of $\rSVL$,
$\{\S,\V,v,T\}^{v\in\{0,1\}^\lambda}$ as follows,

\begin{itemize}[label=\textbullet]
 \item \textsf{Setup}$(1^\lambda, T) \rightarrow \pp$
  It samples a $\rSVL$ sub-family $\{\S,\V,v,T\}_{v\in \{0,1\}^\lambda}$ as the public parameter $\pp$
  from the family $\{\S,\V,v,T\}_{v,\lambda}$, uniformly at random. 
 
 \item \textsf{Eval}$(\pp, x) \rightarrow (y, \bot)$ It takes an input 
 $x\in\mathcal{X}=\{0,1\}^\lambda$, and produces an output $y:=\S^T(x)$. 
 There is no proof, so $\pi=\bot$.
 
 \item \textsf{Verify}$(\pp, x, y, \bot) \rightarrow \{0, 1\}$ 
  It returns $\V(y, T)$. Note that the input $x$ is implicit 
  to the circuit $\V$. Thus, $\verify$ is not independent of $x$. 
\end{itemize}

We prove the correctness, computational soundness and sequentiality of this $\vdf$ 
in Theorem.~\ref{thm:correct}, \ref{thm:sound} and \ref{thm:sequence}.
\end{proof}

\begin{theorem}{\normalfont \textbf{(Correctness).}}\label{thm:correct}
The derived $\vdf$, in Theorem.~\ref{thm:member}, is correct. 
\end{theorem}

\begin{proof}
For the $\rSVL$ instance $(\S,\V,v,T)$, 
$\V(u,T)=1$ if and only if either $u=\S^T(v)$ or $u$ is a false positive
vertex i.e., $\V(u,T)=1$  but $u\ne\S^T(v)$. 
In this $\vdf$, $\eval(x)=\S^T(x)$ and $\verify(x,y,T)=\V(y,T)$. Therefore, 
$\verify(x,y,T)=1$ if and only if either $y=\eval(x,T)$ or $y$ is false positive vertex off
the main line in the $\rSVL$ instance $(\S,\V,x,T)$. Since, the number of false positive
vertices is in $\poly(\lambda)$, the probability that a random vertex $y$ is false
positive is at most $\poly(\lambda)/2^\lambda=\negl(\lambda)$. Therefore, 
$$\Pr[\V(y,T)=1 \land  y=\S^T(x)] \ge  1-\negl(\lambda).$$
\end{proof}

\begin{theorem}{\normalfont \textbf{(Soundness).}}\label{thm:sound}
For the parameters $\lambda\in \Z^+$ and $T=T(\lambda) \in 2^{o(\lambda)}$, let 
$\S : \{0,1\}^\lambda \rightarrow \{0,1\}^\lambda$, $\V : \{0,1\}^\lambda \times \{1,
\ldots, T\}\rightarrow \{0,1\}$ 
and $v\in\{0,1\}^\lambda$. If there exists an adversary $\adv$ that 
breaks the soundness of this permutation 
$\vdf$s with $\X=\Y=\{0,1\}^\lambda$ and the delay parameter $T$ with 
with a non-negligible probability $\epsilon=\epsilon(\lambda)$, 
then there exists another adversary $\adf$ who 
finds false positive vertices for the instances from the $\rSVL$ 
family $\{\S,\V,v,T\}_{v\in\{0,1\}^\lambda}$ with the probability $\epsilon$.
\end{theorem}

\begin{proof}
For the $\rSVL$ instance $(\S,\V,v,T)$, 
$\V(u,T)=1$ if and only if either $u=\S^T(v)$ or $u$ is a false positive
vertex i.e., $\V(u,T)=1$  but $u\ne\S^T(v)$. 
In this $\vdf$, $\eval(x)=\S^T(x)$ and $\verify(x,y,T)=\V(y,T)$. Therefore, 
$\verify(x,y,T)=1$ if and only if either $y=\eval(x,T)$ or $y$ is false positive vertex off
the main line in the $\rSVL$ instance $(\S,\V,x,T)$. 

As the adversary $\adv$ breaks the soundness of this $\vdf$, he must find a $y \ne
\eval(x,T)$ but $\verify(x,y,T)=1$. Therefore, in order to find a false positive vertex
$y$ in the $\rSVL$ instance $(\S,\V,x,T)$, $\adf$ runs $\adv$ on input $(x,T)$. When
$\adv$ outputs $y$, $\adf$ returns $y$ as the false positive vertex. Therefore,

$$\Pr[\adf \text{ wins }]=\Pr[\adv \text{ wins }] =\epsilon(\lambda).$$

\end{proof}

Since, the number of false positive
vertices is in $\poly(\lambda)$, the probability that a random vertex $y$ is false
positive is at most $\poly(\lambda)/2^\lambda=\negl(\lambda)$. Therefore, $\adv$ 
has negligible advantage.

\begin{theorem}{\normalfont \textbf{(Sequentiality).}}\label{thm:sequence}
For the parameters $\lambda\in \Z^+$ and $T=T(\lambda) \in 2^{o(\lambda)}$, let 
$\S : \{0,1\}^\lambda \rightarrow \{0,1\}^\lambda$, $\V : \{0,1\}^\lambda \times \{1,
\ldots, T\}\rightarrow \{0,1\}$ 
and $v\in\{0,1\}^\lambda$. If there exists an adversary $\adv$ that
breaks the sequentiality of this permutation 
$\vdf$s with $\X=\Y=\{0,1\}^\lambda$ and the delay parameter $T$ 
in time $T_\adv=T_\adv(\lambda) <T$ with 
the non-negligible probability $\epsilon=\epsilon(\lambda)$
then there exists another adversary $\adf$ who solves the instances
from the $\rSVL$ family $\{\S,\V,v,T\}_{v \in\{0,1\}^\lambda}$ 
in time $T_\adv+\O(1) <T$ with the probability $\epsilon$, 
\end{theorem}

\begin{proof}
For the $\rSVL$ instance $(\S,\V,v,T)$, 
$\V(u,T)=1$ if and only if either $u=\S^T(v)$ or $u$ is a false positive
vertex i.e., $\V(u,T)=1$  but $u\ne\S^T(v)$. 
In this $\vdf$, $\eval(x)=\S^T(x)$ and $\verify(x,y,T)=\V(y,T)$. Therefore, 
$\verify(x,y,T)=1$ if and only if either $y=\eval(x,T)$ or $y$ 
is false positive vertex off the main line in the $\rSVL$ instance $(\S,\V,x,T)$. 

As the adversary $\adv$ breaks the sequentiality of this $\vdf$, he must find the
$y=\eval(x,T)$ or a $y' \ne \eval(x,T)$ but $\verify(x,y',T)=1$, in time $T_\adv <T$. 
Therefore, in order to solve the $\rSVL$ instance $(\S,\V,x,T)$, $\adf$ runs $\adv$ on input $(x,T)$. 
When $\adv$ outputs $y$, $\adf$ returns $y$ as the solution in time
$T_\adv+\O(1)$. Therefore,

$$\Pr[\adf \text{ wins }]=\Pr[\adv \text{ wins }] =\epsilon(\lambda).$$

\end{proof}

Theorem.~\ref{thm:member} gives rise to permutation $\vdf$. Although, it suffices to
prove that $\rSVL \in \VDF$, deriving a $\vdf$ with $\X \ne \Y$, from $\rSVL$ needs 
a family of injective one-way function $\mathcal{H}=\{H:\X \rightarrow \Y\}$. 

\begin{theorem}{\normalfont \textbf{($\rSVL\in \VDF$, in general).}}\label{thm:memberH}
For the parameters $\lambda\in \Z^+$ and $T=T(\lambda) \in 2^{o(\lambda)}$, let 
$\S : \{0,1\}^\lambda \rightarrow \{0,1\}^\lambda$, $\V : \{0,1\}^\lambda \times \{1,
\ldots, T\}\rightarrow \{0,1\}$ 
and $v\in\{0,1\}^\lambda$. If there exists an adversary $\adv$ that
$x\in\{0,1\}^\lambda$. If there exists a family of $\rSVL$ instances 
$\{\S,\V,v,T\}_{v\in\{0,1\}^\lambda}$ such that each 
instance allows at most polynomially many (i.e., $\poly(\lambda)$) false positive vertices  
then there exists a permutation $\vdf$s with $\X\in \{0,1\}^*$, $\Y=\{0,1\}^\lambda$ 
and the delay parameter $T$, assuming
a family of injective one-way function $\mathcal{H}=\{H:\X \rightarrow \Y\}$. 
\end{theorem}

\begin{proof}
Given any $\lambda\in \Z^+$ and any $T \in 2^{o(\lambda)}$,
we derive a permutation $\vdf$ from a sub-family of $\rSVL$,
$\{\S,\V,v,T\}^{v\in\{0,1\}^\lambda}$ as follows,

\begin{itemize}[label=\textbullet]
 \item \textsf{Setup}$(1^\lambda, T) \rightarrow \pp$
  It samples a $\rSVL$ sub-family $\{\S,\V,v,T\}^{v\in \{0,1\}^\lambda}$ as the public parameter $\pp$
  from the family $\{\S,\V,v,T\}_{v\in\{0,1\}^\lambda}$, uniformly at random. 
 Apart from the $\rSVL$ instance, it also chooses an $H \in \mathcal{H}$, uniformly at
 random.  
 \item \textsf{Eval}$(\pp, x) \rightarrow (y, \bot)$ It takes an input 
 $x\in\{0,1\}^*$, and produces an output $y:=\S^T(H(x))$. 
 There is no proof, so $\pi=\bot$.
 
 \item \textsf{Verify}$(\pp, x, y, \bot) \rightarrow \{0, 1\}$ 
  It returns $\V(y, T)$. Note that the input $H(x)$ is implicit 
  to the circuit $\V$. Thus, $\verify$ is not independent of $x$. 
\end{itemize}

The proofs for the correctness, computational soundness and sequentiality of this $\vdf$ 
are same as Theorem.~\ref{thm:correct}, \ref{thm:sound} and \ref{thm:sequence}.
\end{proof}

As before, first we show that every permutation $\vdf$ reduces to $\rSVL$.

\begin{theorem}{\normalfont \textbf{(Reduction from Permutation $\VDF$ to $\rSVL$).}}
For the parameters $\lambda\in \Z^+$ and $T=T(\lambda) \in 2^{o(\lambda)}$, let 
$(\setup, \eval,\mathsf{Open}, \verify)$ be an interactive permutation $\vdf$ on
 the domain $\{0,1\}^\lambda$. Then
there exists a hard distribution of $\rSVL$ instances
$\{\S,\V,v,T\}_{v\in\{0,1\}^\lambda}$  
that have at most polynomially many (i.e., $\poly(\lambda)$) false positive vertices, such that
$\S : \{0,1\}^\lambda \rightarrow \{0,1\}^\lambda$, $\V : \{0,1\}^\lambda \times \{1,
\ldots, T\}\rightarrow \{0,1\}$ 
and $x\in\{0,1\}^\lambda$. 
\end{theorem}

\begin{proof}
Given any $\lambda\in \Z^+$ and any $T \in 2^{o(\lambda)}$,
we derive a hard distribution of $\rSVL$ instances, $\{\S,\V,v,T\}_{v\in\{0,1\}^\lambda}$ 
from a permutation $\vdf$, as follows,

\begin{description}
\item [The $\S$ circuit] 
In order to design $\S$, we observe that, for every permutation $\vdf$,
$\eval(x,T)=\eval(\eval(x,T-1))$ as $\X=\Y$.
This vertex $u$ is the source in the $\rSVL$ instances. 
Therefore, we keep $\S(u)=\eval(\pp,u,T=1)$. In particular,

\begin{algorithm}
\caption{$\S(u)$ from $\eval$}\label{alg:eval}
\begin{algorithmic}[1]
\STATE $y:=\eval(\pp,u,1)$
\RETURN $y$.
\end{algorithmic}
\end{algorithm}

More generally, $\S^i(v)=\eval(\pp,v,i)$ for all $i \le T$.

\item [The $\V$ circuit] 
We take the advantage of the interactive $\vdf$s in order to design $\V$.
In particular, 
\begin{algorithm}[H]
\caption{$\V(v,T)$ from $\eval$}\label{alg:verify}
\begin{algorithmic}[1]
\STATE $\pi:=\mathsf{Open}(\pp,v,y,T)$
\STATE $w:=\verify(\pp,v,y,T,\pi)$.
\RETURN $w$
\end{algorithmic}
\end{algorithm}

 Since, $\mathsf{Open}$ takes at most $\poly(\lambda, \log
T)$-rounds, $\V$ is efficient.
\end{description}

\end{proof}

Now, the challenge is to reduce any arbitrary $\vdf$ into an hard $\rSVL$ instances.
We can not follow the approach used in~\cite{ChoudhuriRSW}. They label the $\rSVL$
graph with the proofs in the Pietrzak's $\vdf$ using the "proof-merging" technique
described in Sect.~\ref{literature}. This approach does not work for the $\vdf$s that
needs no proof e.g., the isogenie-based $\vdf$~\cite{Feo2019Isogenie}. Labelling the $i$-th node
in the $\rSVL$ graph with $\eval(\pp,x,i)$ for $i \le T$, needs $T$ different computation 
of $\eval$ for each $i$ as $\X \ne \Y$. We tackle this problem in the following theorem
using a special function $f$.

\begin{theorem}{\normalfont \textbf{($\rSVL$ is $\VDF$-hard).}}\label{thm:hardness}
For the parameters $\lambda\in \Z^+$ and $T=T(\lambda) \in 2^{o(\lambda)}$, let 
$(\setup, \eval,\mathsf{Open}, \verify)$ be an interactive $\vdf$ on
 the domain $\X=\{0,1\}^*$ and the range $\Y=\{0,1\}^\lambda$. Then
there exists a hard distribution of $\rSVL$ instances
$\{\S,\V,v,T\}_{v\in\{0,1\}^\lambda}$  
that have at most polynomially many (i.e., $\poly(\lambda)$) false positive vertices, such that
$\S : \{0,1\}^\lambda \rightarrow \{0,1\}^\lambda$, $\V : \{0,1\}^\lambda \times \{1,
\ldots, T\}\rightarrow \{0,1\}$ 
and $x\in\{0,1\}^\lambda$. 
\end{theorem}

\begin{proof}
Given any $\lambda\in \Z^+$ and any $T \in 2^{o(\lambda)}$,
we derive a hard distribution of $\rSVL$ instances, $\{\S,\V,v,T\}_{v\in\{0,1\}^\lambda}$ 
from a $\vdf$, as follows,

First, we define a function $f$ as,
$$f(\eval(\pp,x,T=0),i)=\eval(\pp,x,i) \qquad \forall i \le T.$$

The function $\eval(x,T=0)$ maps $x$ into the range $\Y$.
\begin{description}
\item [The $\S$ circuit] 
We observe that, for every $\vdf$,
$\eval(\pp,v,T)=f(v,T)=f(f(v,T-1))$ with the base case $v=\eval(\pp,x,0)$. 
This vertex $v$ is the source in the $\rSVL$ instances. 
Therefore, we keep $\S(u)=f(\pp,u,T=1)$. In particular, 

\begin{algorithm}
\caption{$\S(u)$ from $\eval$}\label{alg:eval}
\begin{algorithmic}[1]
\STATE $y:=f(\pp,u,1)$
\RETURN $y$.
\end{algorithmic}
\end{algorithm}

More generally, $\S^i(v)=f(\pp,v,i)$ for all $i \le T$.

\item [The $\V$ circuit] 
This circuit is same as in the previous theorem. 
In particular, 
\begin{algorithm}[H]
\caption{$\V(v,T)$ from $\eval$}\label{alg:verify}
\begin{algorithmic}[1]
\STATE $\pi:=\mathsf{Open}(\pp,v,y,T)$
\STATE $w:=\verify(\pp,v,y,T,\pi)$.
\RETURN $w$
\end{algorithmic}
\end{algorithm}

 Since, $\mathsf{Open}$ takes at most $\poly(\lambda, \log
T)$-rounds, $\V$ is efficient.
\end{description}

\end{proof}

\bibliographystyle{splncs04}

\end{document}